\documentclass[runningheads,a4paper]{llncs}

\usepackage{amssymb,amsmath}
\usepackage{graphicx}

\usepackage{url}
\urldef{\mailsa}\path{naldi@disp.uniroma2.it}
\urldef{\mailsb}\path|dacquisto@ing.uniroma2.it|
\newcommand{\keywords}[1]{\par\addvspace\baselineskip
\noindent\keywordname\enspace\ignorespaces#1}

\begin{document}

\mainmatter  

\title{Differential privacy for counting queries: can Bayes estimation help uncover the true value?}

\titlerunning{Bayes estimation in counting queries}

%
%
\author{Maurizio Naldi%
\and Giuseppe D'Acquisto}
\authorrunning{M. Naldi - G. D'Acquisto}

\institute{Universit\`{a} di Roma Tor Vergata\\Department of Computer Science and Civil Engineering\\
Via del Politecnico 1, 00133 Roma, Italy\\
\mailsa\\
\mailsb\\
}

%
%

\toctitle{Lecture Notes in Computer Science}
\tocauthor{Authors' Instructions}
\maketitle

\begin{abstract}
Differential privacy is achieved by the introduction of Laplacian noise in the response to a query, establishing a precise trade-off between the level of differential privacy and the accuracy of the database response (via the amount of noise introduced). Multiple queries may improve the accuracy but erode the privacy budget. We examine the case where we submit just a single counting query. We show that even in that case a Bayesian approach may be used to improve the accuracy for the same amount of noise injected, if we know the size of the database and the probability of a positive response to the query.
\end{abstract}

\section{Introduction}
Privacy concerns are an ever-growing issue in a world where more and more information is created, circulated, and shared. In addition to the personal data that users share voluntarily, many apps grab users' personal data with their formal consent, though the users themselves are not fully aware of the implications \cite{wang2011} \cite{zhou2012}. 

Threats to privacy are not limited to databases where queries about individuals are possible, but extend to statistical databases, where queries are allowed just for aggregate data \cite{shoshani1982}. Several approaches have been proposed to preserve privacy in that context \cite{adam1989} \cite{agrawal2000}. The notion of differential privacy has emerged, which at the same provides a definition of quantifiable privacy and a way to guarantee it (see the introducing paper by Dwork \cite{dwork2006differential}, the subsequent survey in \cite{dwork2008}, and its re-examination in 2011 \cite{dwork2011}). In differential privacy, the level $\epsilon$ of privacy guaranteed to an individual is measured through the extent to which its inclusion in a database changes the database response to a query.

It has been shown that the properties of differential privacy may be achieved by using a specialization of the output perturbation technique already considered in early works on the subject \cite{adam1989}: the addition of Laplacian noise to the database true response. However, the disclosure of aggregate data under the differential privacy scheme is not immune from problems. For example, through a synthetic dataset McClure and Reiter have questioned the use of the differential privacy level as a measure of the statistical disclosure risk, showing that the probability of an intruder uncovering true values may be significant even for high levels of differential privacy \cite{mcclure2012}. On the other hand, allowing multiple queries, which worsens the level of differential privacy and erodes the privacy budget, may require the addition of large quantities of noise \cite{sarathy2011} \cite{Heffetz2014}. In the presence of multiple queries, optimization of the noisy response has been sought after to preserve differential privacy while maintaining an adequate accuracy for the database response to be useful \cite{Li2010}. 

A central issue in differential privacy is therefore the trade-off between privacy and accuracy, with the use of multiple queries as a tool to achieve better accuracy but degrading the level of privacy. No attention has been specifically devoted instead to the case of a single query. Since established statistical estimation approaches (based, e.g., on the maximum likelihood principle, moment matching, or order statistics) require a sample, i.e., the output of multiple queries, it looks like the single noisy response from the database cannot be improved upon, with accuracy and privacy being tied in an immutable relationship.

In this paper, we consider instead the case of a single noisy response to a counting query and examine whether its accuracy can be improved.  We adopt a Bayesian approach, where the noisy response is corrected with the aid of two additional information: the size of the database and the expected fraction of positive responses (which does not constrain the actual response to the query). The Bayesian approach has been proposed in the past to reconstruct the distribution of the data after multiple queries \cite{agrawal2000}, while we focus on a single query and aim at the true value of the database response rather than the whole distribution.

We show that, for the same quantity of noise addition, the Bayesian approach achieves a lower average error than the simple noisy response and provides an estimate closer to the true value response in most cases. While this is achieved under the differential privacy requirements, the accuracy improvement represents a clear sign that the overall level of privacy may be vulnerable even for a single query under noise addition. We also provide an expression for the probability that the addition of Laplacian noise in counting queries leads to patently wrong responses, i.e., outside the valid range represented by the size of the database.

The paper is organized as follows. We describe the noise addition paradigm and its limitations in Section \ref{distortion}. We introduce Bayes estimation in Section \ref{estimation} and then provide both the evaluation framework and the simulation-based results in Section \ref{evaluation}.

\section{Value distortion}
\label{distortion}
\textit{Value distortion} is one of the methods in the general class of \textit{output perturbation} techniques to preserve privacy when responding to queries. Alternative methods are \textit{value-class membership} and \textit{input perturbation} \cite{agrawal2000} \cite{dwork2011}. We are interested in the use of this method for counting queries when differential privacy is aimed at. In this section, we first define counting queries, and then describe the characteristics of value distortion to achieve differential privacy.

We consider a database $D$ as a collection of $n$-tuples (the records) from some abstract domain $X$, so that $D\in X^{n}$ and $n=\vert D \vert$ (the size of the database). The set $X$ represents the set of all possible records. Here we do not pose any constraint on the nature of $X$: its elements could be strings of text, dates, Boolean or real values. We consider a property $S$, which each record may satisfy or not. Again, we do not pose any constraint on the type of property: we assume just that for any record the database owner can assess unambiguously if the record of interest satisfies the property or not. The result of a query concerning a record may therefore be expressed by a Boolean output. We can therefore set the following definition.
\begin{definition} A counting query $Q_{\psi}$ to a database $D$ concerning a predicate $\psi \quad :\quad X \rightarrow \{ 0,1\} $ is defined as the number of records in $D$ satisfying the predicate $\psi$\\
\begin{center}
 $Q_{\psi}(D) = \sum_{x\in D} \psi (x)$.
 \end{center} \end{definition}

A straightforward consequence of the definition is that $Q_{\psi}\in \{0,1,2,\ldots ,n\}$.

Under a value distortion approach, the database owner does not respond to a counting query by providing the true value $Q_{\psi}$, but adds some noise. The introduction of noise allows the database owner to provide a useful approximate response, while not disclosing the true value. We define therefore the noisy query
\begin{definition} A noisy query $Q_{\psi}^{*}$ to a database $D$ concerning a predicate $\psi \quad :\quad X \rightarrow \{ 0,1\} $ is defined as the sum of the number of records in $D$ satisfying the predicate $\psi$ and a random variate $R$\\
\begin{center}
 $Q_{\psi}^{*}(D) = R + \sum_{x\in D} \psi (x)$.
 \end{center} \end{definition}

If we let $R$ be a continuous random variable and do not pose any restriction on its value, we see that this alters the range of the noisy query with respect to the true value. We now have $Q_{\psi}^{*} \in \mathbb{R}$ and in particular we may end up with either $Q_{\psi}^{*}<0$ or $Q_{\psi}^{*} > n$. In those cases, the response delivered by the database owner is patently wrong.

The choice of the probability density function of the noise is critical. In \cite{agrawal2000} the uniform and the Gaussian distributions have been considered. Dwork has proven that for a mechanism that adds noise bounded by $E$ there exists an adversary that can reconstruct the database to within $4E$ positions \cite{dwork2011}. Probability distributions with a bounded domain should therefore be excluded from the set of possible choices. 

Over the past years the concept of differential privacy has emerged, with a tight relationship with the choice of distribution for value distortion. According to the general definition, a randomized function $\mathcal{K}$ gives $\epsilon$-differential privacy if, for all neighbouring databases $D_{1}$ and $D_{2}$ (i.e., differing just for one record), and all $U\subseteq \textit{Range}(\mathcal{K})$, we have $\mathbb{P}[\mathcal{K}(D_{1})\in U]\leq e^{\epsilon}\mathbb{P}[\mathcal{K}(D_{2})\in U]$ \cite{dwork2008}. As the privacy level $\epsilon$ gets closer to 0, the privacy guaranteed to an individual grows, since its inclusion in the database does not change appreciably the response of the database.  Dwork et alii have shown that adding Laplace-distributed noise to a counting query (the randomized mechanism $\mathcal{K}$) achieves $\epsilon$-differential privacy if the parameter of the Laplace distribution is $b=\Delta f/\epsilon$, with $\Delta f$ representing the maximum difference in the value of the response when exactly one input to the database is changed  \cite{dwork2006} \cite{dwork2011}. For a counting query we have $\Delta f=1$ \cite{dwork2011}  \cite{sarathy2011}, since the addition of a record may modify the response by 1 at most. As a consequence we have $b=1/\epsilon$, and the standard deviation of the Laplace distribution is $\sqrt{2}/\epsilon$. There is therefore a tight relationship between the amount of noise injected and the level of differential privacy: larger values of the standard deviation of noise lead to more differential privacy (i.e., lower values of $\epsilon$). 

The Laplacian density function of $R$ with parameter $b$ (and variance $2b^{2}$) or, equivalently, with $\epsilon$-differential privacy is then
\begin{equation}
 f_{R}(z)=\frac{1}{2b}\exp(-\vert z \vert/b) = \frac{\epsilon}{2}\exp (-\epsilon\vert z \vert).
 \end{equation} 

Under the differential privacy conditions stated above, it can be easily shown that the probability $P^{*}$ of out-of-range values for the noisy response when the true response is $a$ is a function of the true value itself
\begin{equation}
\label{lapout}
P^{*} = \mathbb{P}[(Q_{\psi}^{*}<0 \bigcup Q_{\psi}^{*}>n \quad \vert \quad Q_{\psi}=a ] = \frac{e^{-a/b}+e^{(a-n)/b}}{2}
\end{equation}
While the noisy response is always non-exact, the chance that an out-of-range value is returned may lead to a patently wrong response. It is therefore desirable that the probability $P^{*}$ be as low as possible. Actually, for the case of Laplacian noise we see that it may be quite large. We can formulate the following bounds
\begin{theorem} If Laplacian noise is added to achieve $\epsilon$-differential privacy, the probability $P^{*}$ of out-of-range values is always $P^{*}\le\frac{1+e^{-\epsilon n}}{2}$. The maximum value of $P^{*}$ is achieved when the true response is $a=0$ or $a=n$.\end{theorem}
\begin{proof} We first recall that under $\epsilon$-differential privacy with Laplacian noise, the probability density function of the noise $R$ is  $f_{R}(z)=\frac{\epsilon}{2}\exp(-\epsilon \vert z \vert)$.  According to Equation (\ref{lapout}) the probability of out-of-range values is then $P^{*} = \frac{e^{-\epsilon a}+e^{\epsilon(a-n)}}{2}$. In order to prove the theorem, we first compute the value of $P^{*}$ at either end of the valid interval ($a=0$ or $a=n$) and then prove that the value of $P^{*}$ inside the valid interval is always lower.\\
By inserting either $a=0$ or $a=n$, we obtain immediately $P^{*}=\frac{1+e^{-\epsilon n}}{2}$. \\
We can now see for which values of the true response the probability $P^{*}$ is increasing. Since $a\in \mathbb{N}$, that condition can be expressed as $P^{*}(a+1)>P^{*}(a)$, i.e. 
\begin{equation}
\frac{e^{-\epsilon (a+1)}+e^{\epsilon (a+1) - \epsilon n}}{2} > \frac{e^{-\epsilon a}+e^{\epsilon a -\epsilon n}}{2}.
\end{equation}
After some algebraic manipulation, this condition can be expressed as
\begin{equation}
e^{2\epsilon a}>e^{\epsilon n - \epsilon} \rightarrow a>\frac{n-1}{2}.
\end{equation}

We see therefore that the function $P^{*}(a)$ is first decreasing and then increasing: the maximum value is achieved when either $a=0$ or $a=n$, and is
\begin{equation}
\max P^{*} = \frac{1+e^{-\epsilon n}}{2},
\end{equation}
which completes the proof.  \qed 
  \end{proof}
Even though limited to the extreme ends of the valid interval, the probability of getting a non-valid response may be pretty high, even larger than 50\%. 

On the other hand, the minimum value of the probability of out-of-range values is obtained when the true value is $a = (n-1)/2$, which is an integer just when the number of records in the database is odd. By replacing this value in Equation (\ref{lapout}), with $b=1/\epsilon$, after some algebraic manipulation we obtain
\begin{equation}
\min P^{*} \simeq \frac{e^{-\frac{\epsilon}{2}(n-1)}+e^{-\frac{\epsilon}{2}(n+1)}}{2}.
\end{equation}

\section{Estimation of the true value}
\label{estimation}
In Section \ref{distortion}, we have seen how privacy can be preserved by perturbing the true response to a counting query through noise addition. In particular, we have seen that the addition of Laplacian noise with parameter $b=\epsilon$ leads to $\epsilon$-differential privacy. However, the response provided may be far from the true one. We wish to get instead a response as close as possible to the true value. In this section, we consider a method, based on a Bayesian approach, to correct the noisy response and get a better estimate of the true value. We contrast it with what we call the naive estimator. We assume anyway that the estimate is based on a \textbf{single} noisy response. Unlike other approaches (as well as typical estimation methods), we do not consider estimates based on a sample of values, i.e., the repetition of queries to obtain several noisy responses and improve the estimate by averaging or other statistical procedures. In the following, we also assume that the added noise follows a Laplacian distribution, with parameter $b=1/\epsilon$ to achieve $\epsilon$-differential privacy.

\subsection{The naive estimator}
A simple line of action is to consider the noisy response as an estimate of the true value. We call this the \textit{naive estimator} and represent it by the symbol $\hat{Q}_{\textrm{naive}}$. Its quality depends on the amount of differential privacy that has been injected. Such a naive estimator is unbiased, since the  expected value of the added noise is zero. But its variance, which is also the variance of the resulting noisy response, is \cite{dodge2008}
\begin{equation}
\mathrm{Var}[\hat{Q}_{\textrm{naive}}] = \mathrm{Var}[R] = 2b^{2} = \frac{2}{\epsilon ^{2}}.
\end{equation}

\subsection{Bayes estimator}
\label{bayesdef}
The generation of a noisy response, as described in Section \ref{distortion} starts with the observation of the true response, but the latter is the result of the query on a database whose entries, though known to the database owner, are actually generated by a random process, since the characteristics of the database population are not perfectly predictable in advance. The true response can therefore be considered as a random variable, so that the distribution of the noisy response is a biased Laplacian conditioned on the true value of the response
\begin{equation}
f_{Q_{\psi}^{*}}(y \vert Q_{\psi}=a) = \frac{\epsilon}{2}e^{-\epsilon\vert y-a \vert}.
\end{equation}
The revelation of the noisy response provides further information on the actual true value, though masked by the addition of noise.

In this context, we can apply Bayes' theorem, where the marginal distribution of the true value plays the role of prior distribution, and we are interested in obtaining the posterior distribution of the true value updated after the noisy response
\begin{equation}
\mathbb{P}[Q_{\psi}=k \vert Q_{\psi}^{*}=y] = \frac{\mathbb{P}[Q_{\psi}=k]f_{Q_{\psi}^{*}}(y \vert  Q_{\psi}=k)}{\sum_{i=0}^{n}\mathbb{P}[Q_{\psi}=i]f_{Q_{\psi}^{*}}(y \vert Q_{\psi}=i)}
\end{equation}

The resulting estimate for the true value after a \textbf{single} observation of the noisy response can therefore be formulated as
\begin{equation}
\label{bayes-1}
\mathbb{E}[Q_{\psi}\vert Q_{\psi}^{*}=y] = \frac{\sum_{k=0}^{n}k\mathbb{P}[Q_{\psi}=k]f_{Q_{\psi}^{*}}(y \vert Q_{\psi}=k)}{\sum_{i=0}^{n}\mathbb{P}[Q_{\psi}=i]f_{Q_{\psi}^{*}}(y \vert Q_{\psi}=i)}
\end{equation}

In Equation (\ref{bayes-1}), we know the distribution of the observed noisy response conditioned on the true value (a shifted Laplacian), and we assume to know the prior distribution of the true value. A reasonable assumption is to assume that the probability that a generic record satisfies the property $S$ for which we submit a query is a known value $p$, and that all the records are statistically independent of one another. This assumption results in a binomial distribution for the number of records satisfying the property $S$, as already assumed in \cite{mcclure2012}. The final form of the estimator is then 
\begin{equation}
\label{bayes-2}
\begin{split}
\hat{Q}_{\textrm{Bayes}}=\mathbb{E}[Q_{\psi}\vert Q_{\psi}^{*}=y] &= \frac{\sum_{k=0}^{n}k{n \choose k}p^{k}(1-p)^{n-k}\frac{1}{2}\epsilon e^{-\epsilon\vert y-k\vert}}{\sum_{i=0}^{n}{n \choose i}p^{i}(1-p)^{n-i}\frac{1}{2}\epsilon e^{-\epsilon\vert y-i\vert}}\\
&=\frac{\sum_{k=0}^{n}k{n \choose k}p^{k}(1-p)^{n-k} e^{-\epsilon\vert y-k\vert}}{\sum_{i=0}^{n}{n \choose i}p^{i}(1-p)^{n-i} e^{-\epsilon\vert y-i\vert}}.
\end{split}
\end{equation}
We recall that the additional information that the Bayesian estimator requires with respect to the naive estimator is the size $n$ of the database (such a possibility is explicitly mentioned as typical in the census case in \cite{dwork2011} and in the mean salary case in \cite{Heffetz2014}) and the probability that a generic records satisfies the property considered (i.e., the expected fraction of records satisfying that property). As to the latter requirement, it is to be noted that knowing that probability is not very significant to predict the actual value of the query, since the uncertainty related to the binomial generation mechanism is much larger than that due to noise addition. For example, in a database of 10,000 records, where the probability of the property $S$ is 0.3, the $1\sigma$ confidence interval for the query response has a width of $2\sqrt{np(1-p)}=91.6$, while the width of the corresponding interval for the Laplace noise with a $0.1$-differential privacy is a much lower $2\sqrt{2}/\epsilon = 28.2$.

\section{Which estimator is better?}
\label{evaluation}
In Section \ref{estimation}, we have considered two estimators of the true value sought after by the user: the naive estimator and the Bayes estimator. The naive estimator is the noisy response provided by the database owner, which is calibrated so as to guarantee $\epsilon$-differential privacy. The Bayes estimator corrects the noisy response through the only knowledge of the database size and the expected fraction of records satisfying the property the query refers to. If the Bayes estimator achieves a better accuracy, the privacy vs. accuracy trade-off imposed by noise addition is modified to get a better accuracy. In this section, we compare the two estimators to investigate that possibility. We first define two metrics that allow to compute the error suffered by the two estimators, and then provide the error performance as evaluated through simulation.

\subsection{Error metrics}
\label{metdef}
We consider two metrics to quantify the respective accuracy of the two estimators:
\begin{itemize}
\item Average error;
\item Probability of lower error;
\end{itemize}

The first metric is defined separately for the two estimators as follows:
\begin{equation}
\begin{split}
\overline{M}_{\textrm{naive}}&=\mathbb{E}\vert Q_{\psi} - \hat{Q}_{\textrm{naive}} \vert,\\
\overline{M}_{\textrm{Bayes}}&=\mathbb{E}\vert Q_{\psi} - \hat{Q}_{\textrm{Bayes}} \vert,
\end{split}
\end{equation}
where the expectation is taken over the range of possible values of the true response. The better estimator is of course that having the lower average error.

The second metric does not consider the exact differences in value between the error of the two estimator, but considers just the better in a generic outcome of the query:
\begin{equation}
\mathbb{P}[M_{\textrm{Bayes}}<M_{\textrm{naive}}].
\end{equation}
Since it is defined with reference to the Bayes estimator (see the sign of the above inequality), the latter is a better estimator if that probability is larger than 0.5.

It is possible that an estimator, though being better than the other most of the times, is so by a tiny amount while being extremely worse for a small fraction of times. The two metrics are therefore complementary.

\subsection{Simulation-based comparison}
In this section, we provide a comparison of the two estimators using the metrics defined in Section \ref{metdef}. All the results provided in this section have been obtained through MonteCarlo simulation with 100,000 simulation runs for each case, with the exception of the average error for the naive estimator, which can be easily derived analytically. Two different values have been considered for the database population: 100 and 1000. The probability that the response to the query is positive for a generic record has been set equal to 0.3, so that the expected result of the query is respectively 30 and 300. As to the level of differential privacy, Heffetz and Ligett have pointed out in \cite{Heffetz2014} that its setting (being a normative issue) is still an open issue. In the same paper, they consider two values: 0.1 and 1. Here we have considered a range encompassing those two values.

Let's consider first the average error. As hinted above, for the naive estimator the average error is easily computed as
\begin{equation}
\label{avenoise}
\begin{split}
\overline{M}_{\textrm{naive}}&=\mathbb{E}\vert Q_{\psi} - \hat{Q}_{\textrm{naive}} \vert = \mathbb{E}[\vert R \vert]\\ 
&= \int_{-\infty}^{+\infty}\frac{\vert z \vert \epsilon}{2}e^{-\epsilon\vert z \vert}dz = \int_{0}^{+\infty}z\epsilon e^{-\epsilon z}dz = \frac{1}{\epsilon}.
\end{split} 
\end{equation}

The performance of the two estimators depend on the amount of noise and the characteristics of the database population (its size $n$ and its expected response to the query $p$). In \figurename~\ref{fig:Err-1} and \figurename~\ref{fig:Err-2}, we examine the impact of the first two quantities, when $p=0.3$.
\begin{figure}[htbp]
\begin{center}
  \includegraphics[width=.65\columnwidth]{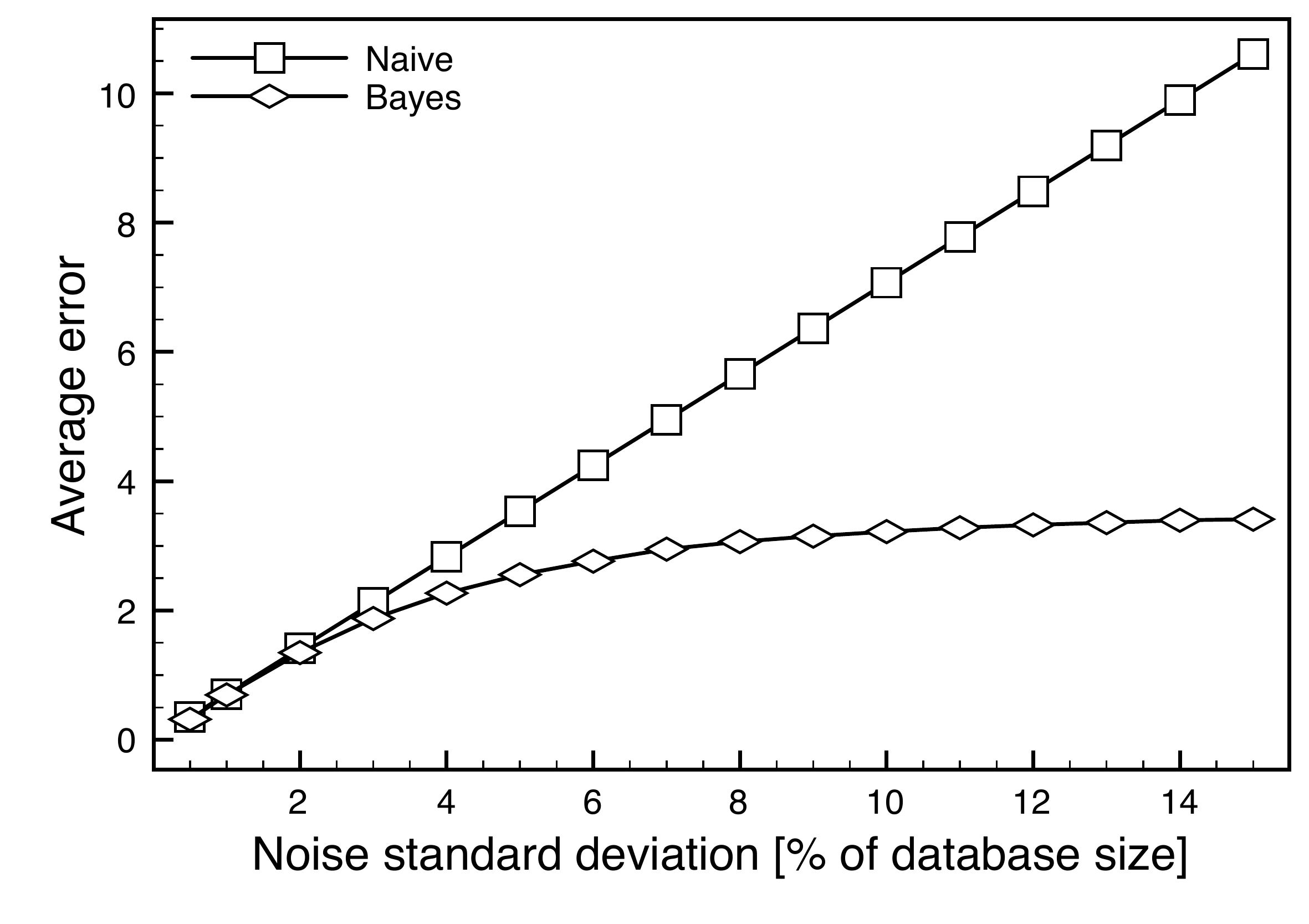}
    \caption{Impact of noise on the average error ($n=100$)}
    \label{fig:Err-1}
\end{center}
\end{figure}
\begin{figure}[htbp]
\begin{center}
  \includegraphics[width=.65\columnwidth]{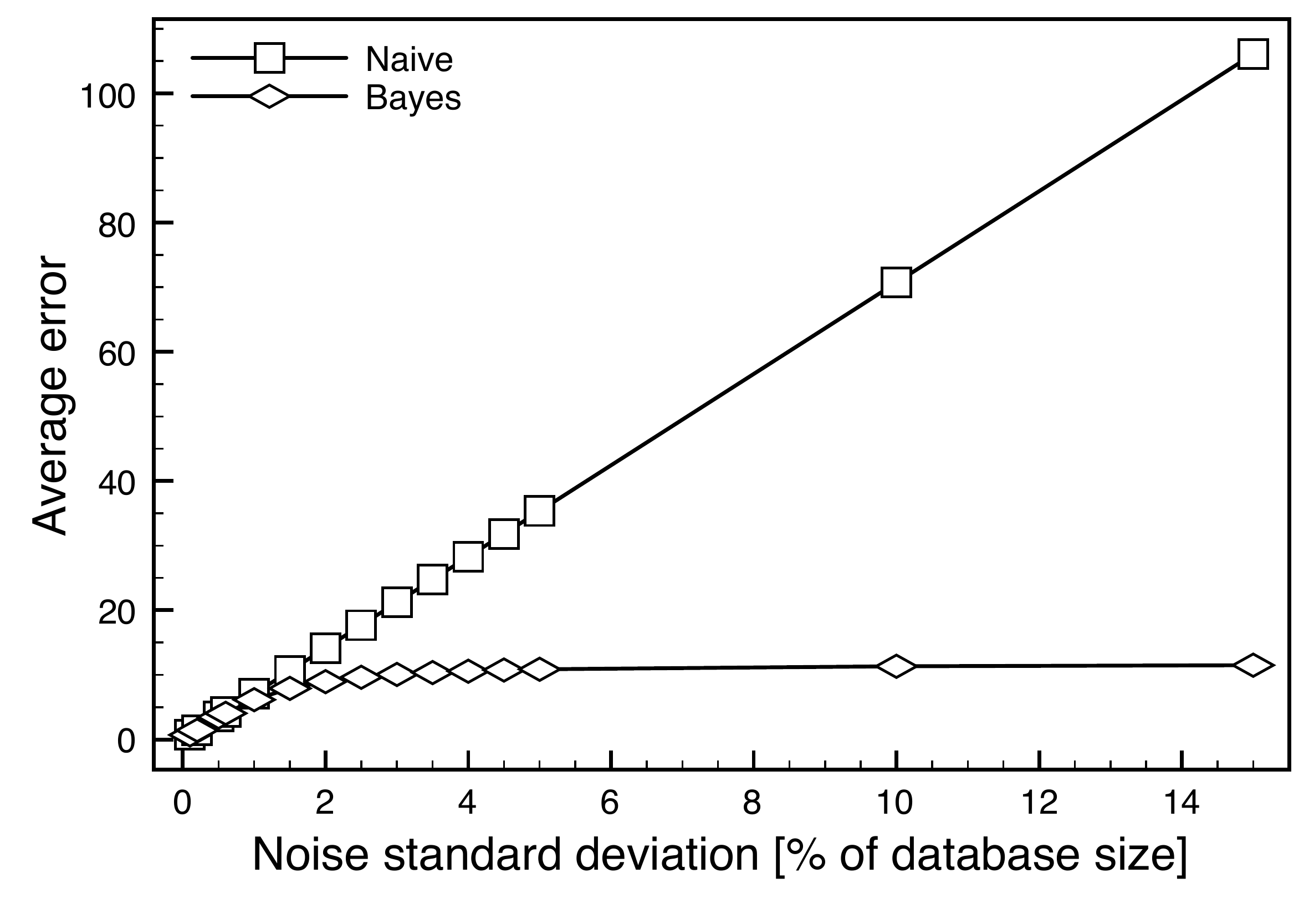}
    \caption{Impact of noise on the average error ($n=1000$)}
    \label{fig:Err-2}
\end{center}
\end{figure}

When the amount of noise injected is small, the performances of the two estimators are very close. But the error of the naive estimator grows linearly with the standard deviation of the noise. Instead, the error of the Bayes estimator levels off, though continuing a very slow growth. For all practical purposes the average error of the Bayes estimator is bounded and always lower than that attained by the naive estimator.

In \figurename~\ref{fig:Err-3} and \figurename~\ref{fig:Err-4}, we plot the average error as a function of the level $\epsilon$ of differential privacy (with the $x$-axis reversed, so that  traversing the graph from left to right brings us from lesser privacy to greater privacy); it is just a recalibration of the $x$-axis in terms of privacy rather than noise amount, since the standard deviation of noise is $\sqrt{2}/\epsilon$. We can formulate the same conclusions as before, but expressed in term of the level of privacy: the better performance of the Bayes estimator stand up in particular as a higher level of privacy is desired. 
\begin{figure}[htbp]
\begin{center}
  \includegraphics[width=.65\columnwidth]{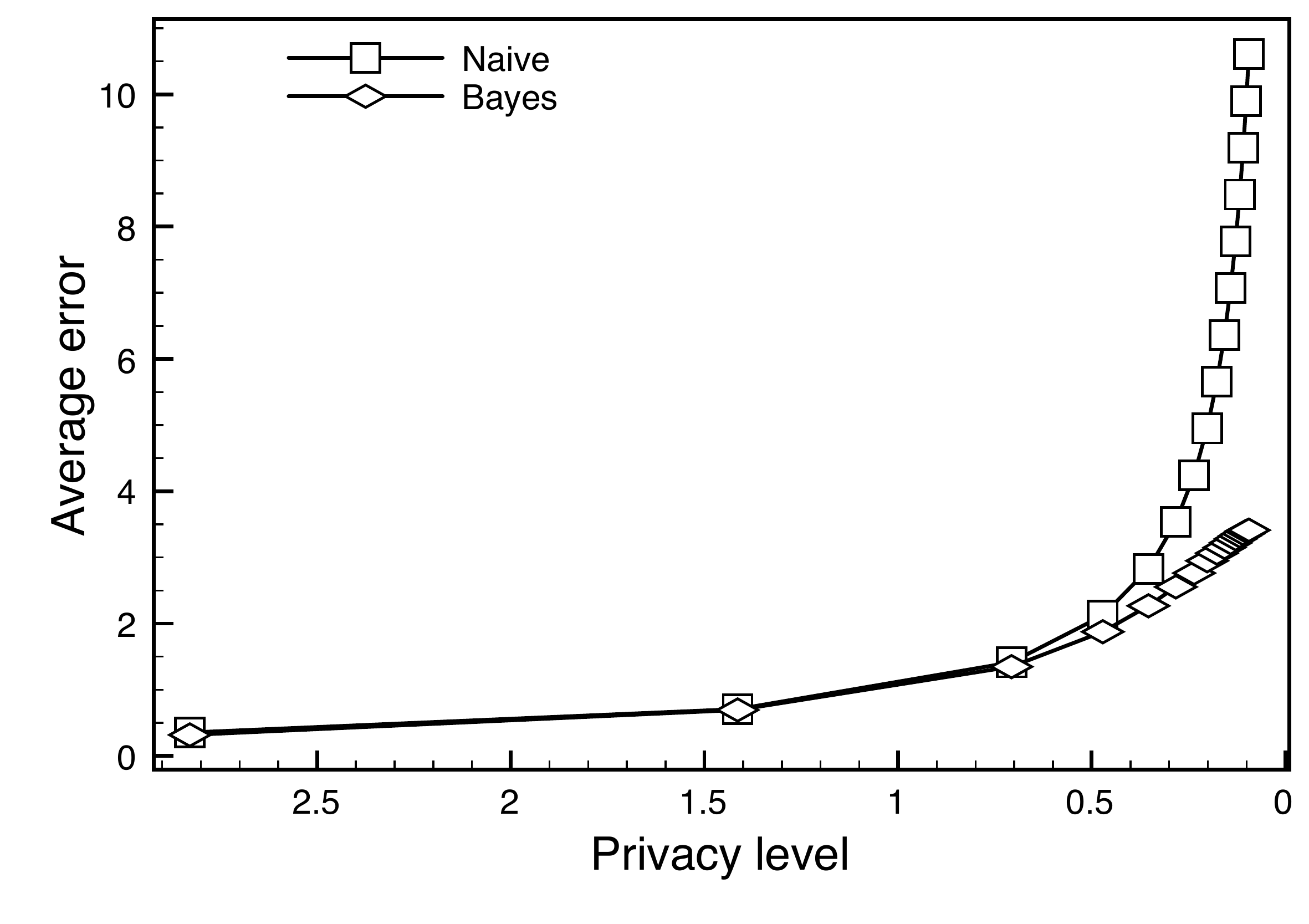}
    \caption{Impact of privacy level on the average error ($n=100$)}
    \label{fig:Err-3}
\end{center}
\end{figure}
\begin{figure}[htbp]
\begin{center}
  \includegraphics[width=.65\columnwidth]{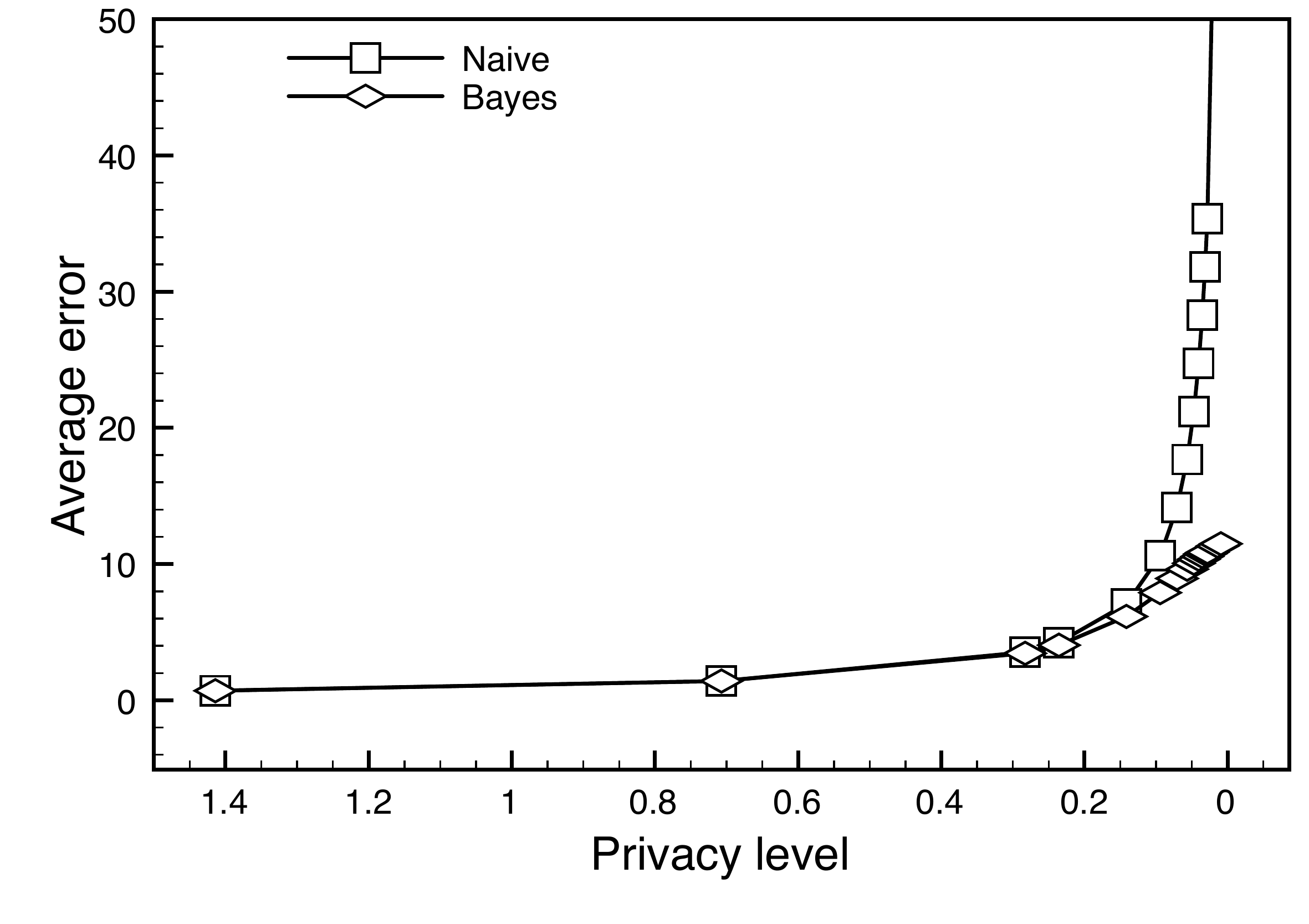}
    \caption{Impact of privacy level on the average error ($n=1000$)}
    \label{fig:Err-4}
\end{center}
\end{figure}

After this first round of results, we can already see that, under the same level of differential privacy, the Bayes estimator achieves a better trade-off between privacy and accuracy than the noisy response only.

We examine now the behaviour of the two estimators with respect to the third quantity: the probability $p$ that the property $S$ is satisfied (we call it the predicate probability, since it represents the probability that the predicate is true for a generic record). In \figurename~\ref{fig:Err-5}, we plot the average error for a differential privacy level $\epsilon = 0.1$. Since the average error of the naive estimator is equal to $1/\epsilon$ as per Equation (\ref{avenoise}), the performance curve of that estimator is a straight horizontal line (here equal to $1/0.1=10$) and depends neither on the predicate probability nor on the database size. Instead, for the Bayes estimator, the result is expected to depend on the variance of the binomially-distributed variable that represents the true response to the query. Since that variance is $np(1-p)$, i.e., a quadratic function of the predicate probability that attains its maximum at $p=0.5$, we observe the same behaviour for the error curves in \figurename~\ref{fig:Err-5}: the average error is minimum when the probability $p$ is very close to either 0 or 1. 
\begin{figure}[htbp]
\begin{center}
  \includegraphics[width=.65\columnwidth]{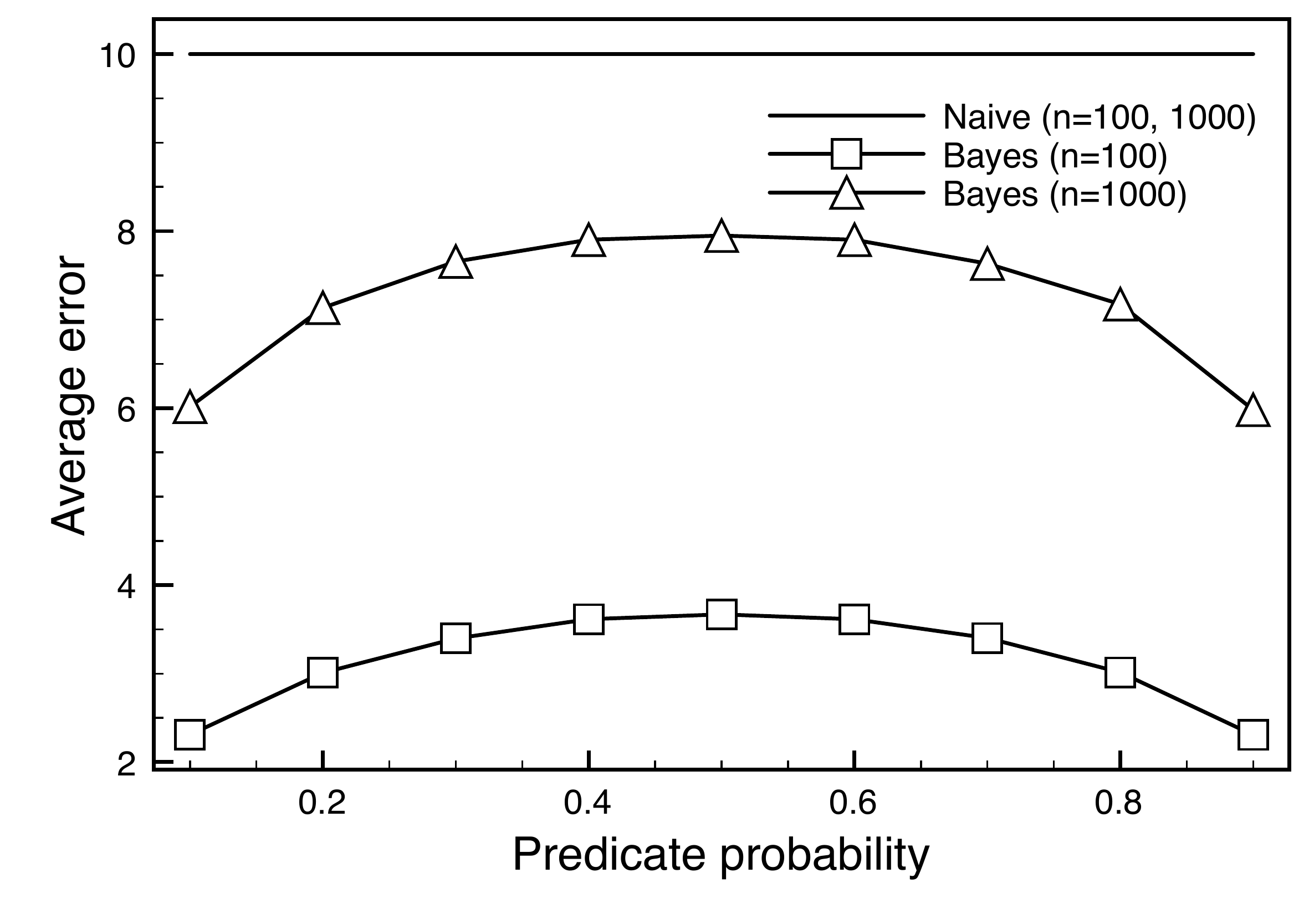}
    \caption{Impact of predicate probability on error}
    \label{fig:Err-5}
\end{center}
\end{figure}

Now we turn to the second performance metric: the probability that the Bayes estimator is better than the naive estimator, as defined in Section \ref{bayesdef}. We recall that that probability has to be larger than 50\% for the Bayes estimator to be actually better than its naive competitor.

We report the results in \figurename~\ref{fig:Perr-1} and \figurename~\ref{fig:Perr-2} respectively versus the noise standard deviation and the level $\epsilon$ of differential privacy. These results have been obtained through MonteCarlo simulation with 100,000 simulation runs. The Bayes estimator is always better than the naive one. The gap between the two increases with the amount of noise injected (or, equivalently, with the level of privacy required). The two estimators have close performances just when $\epsilon$ is so high (low amount of noise injected) that the error is anyway very low under both estimators: when the level of differential privacy is $\epsilon > 1$, the average error for the two estimators follows the inequality $\overline{M}_{\textrm{Bayes}} < \overline{M}_{\textrm{naive}} < 1$, even below the granularity of the counting process.

\begin{figure}[htbp]
\begin{center}
  \includegraphics[width=.65\columnwidth]{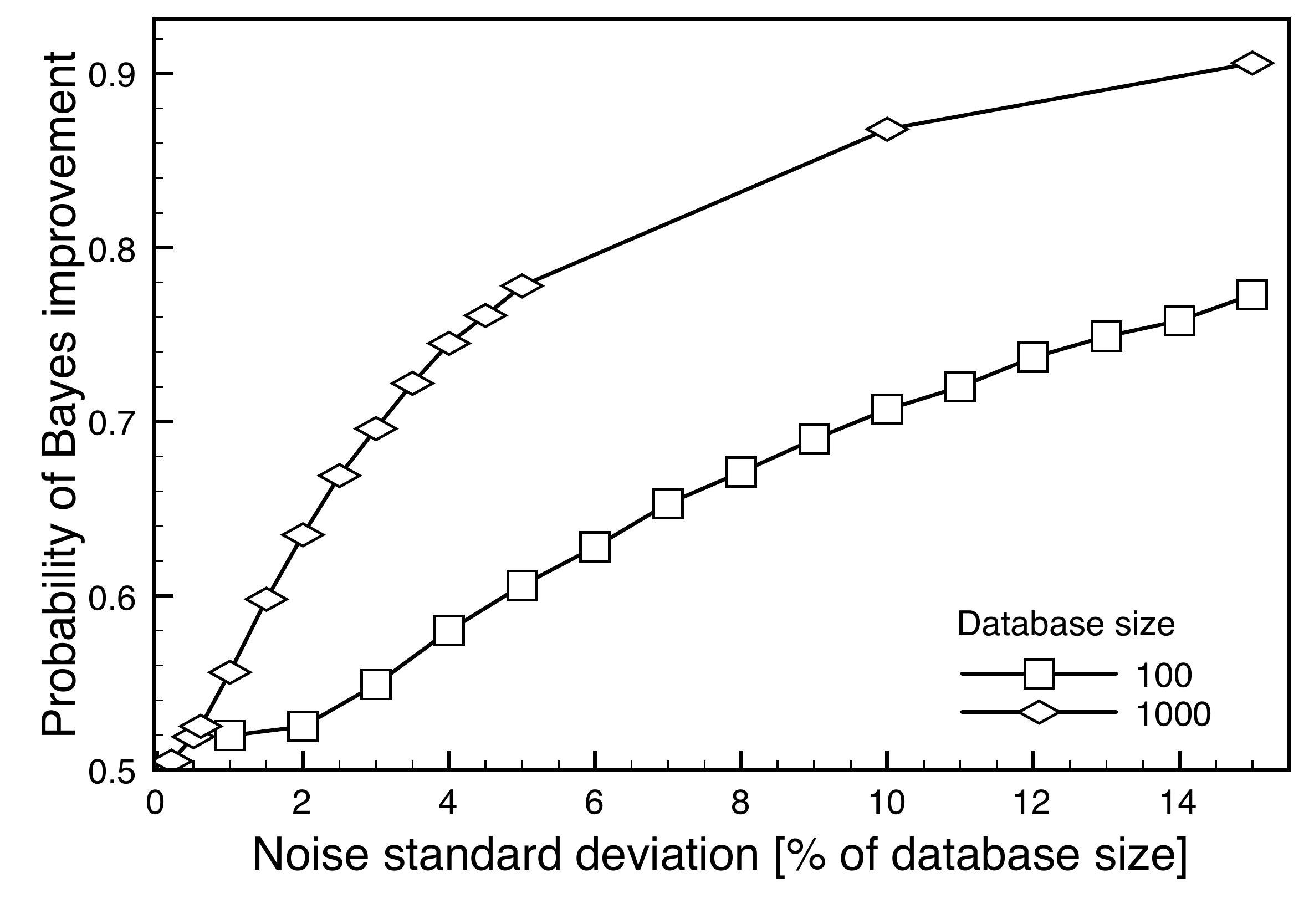}
    \caption{Probability of Bayes improvement vs noise}
    \label{fig:Perr-1}
\end{center}
\end{figure}
\begin{figure}[htbp]
\begin{center}
  \includegraphics[width=.65\columnwidth]{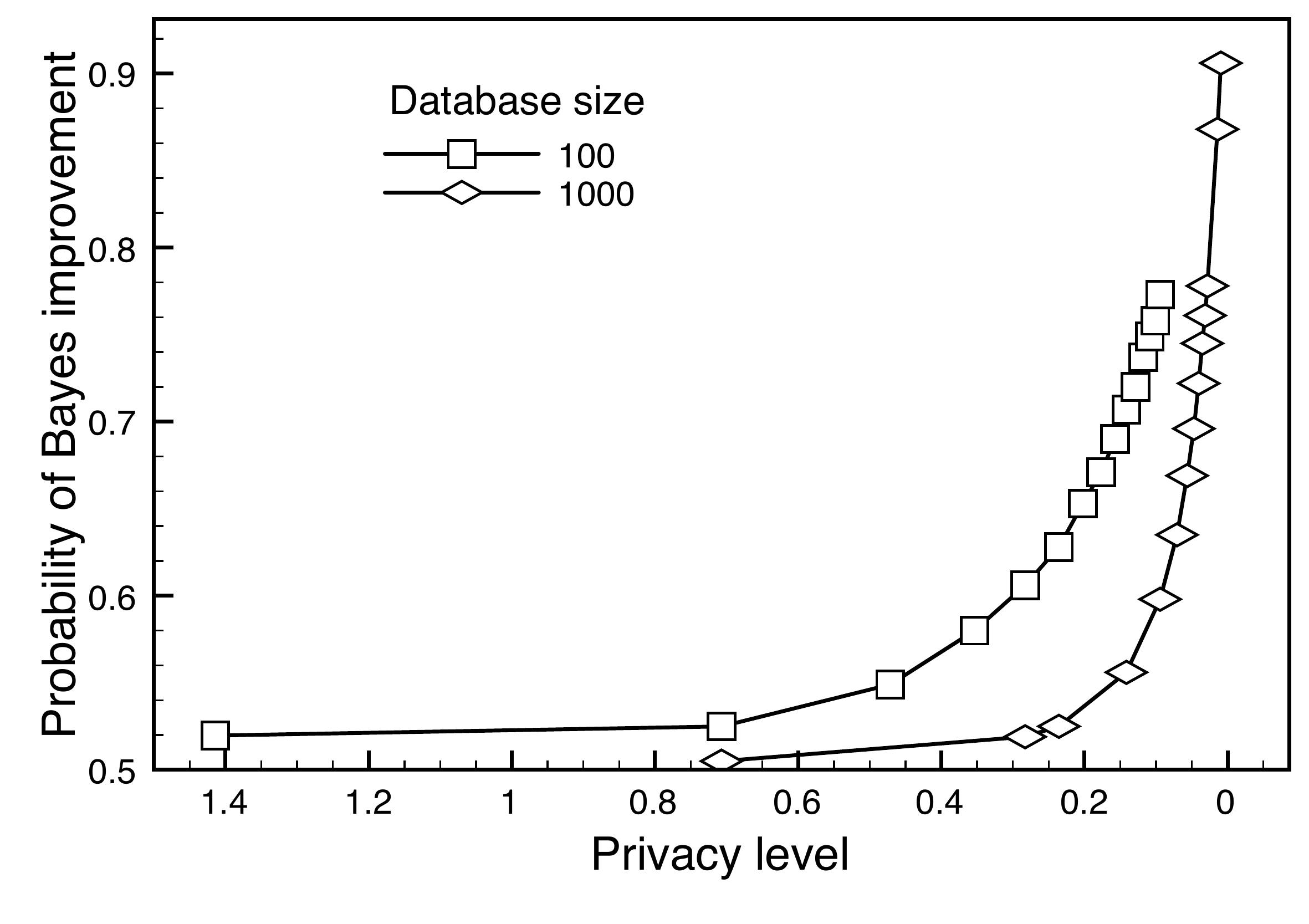}
    \caption{Probability of Bayes improvement vs privacy level}
    \label{fig:Perr-2}
\end{center}
\end{figure}

As to the impact of the predicate probability, we see a similar behaviour as observed in \figurename~\ref{fig:Err-5}. We report the results in \figurename~\ref{fig:Perr-3}, where the performance of the Bayes estimator is the worst at midrange, where the variance of the binomially-distributed true outcome is the largest.

\begin{figure}[htbp]
\begin{center}
  \includegraphics[width=.65\columnwidth]{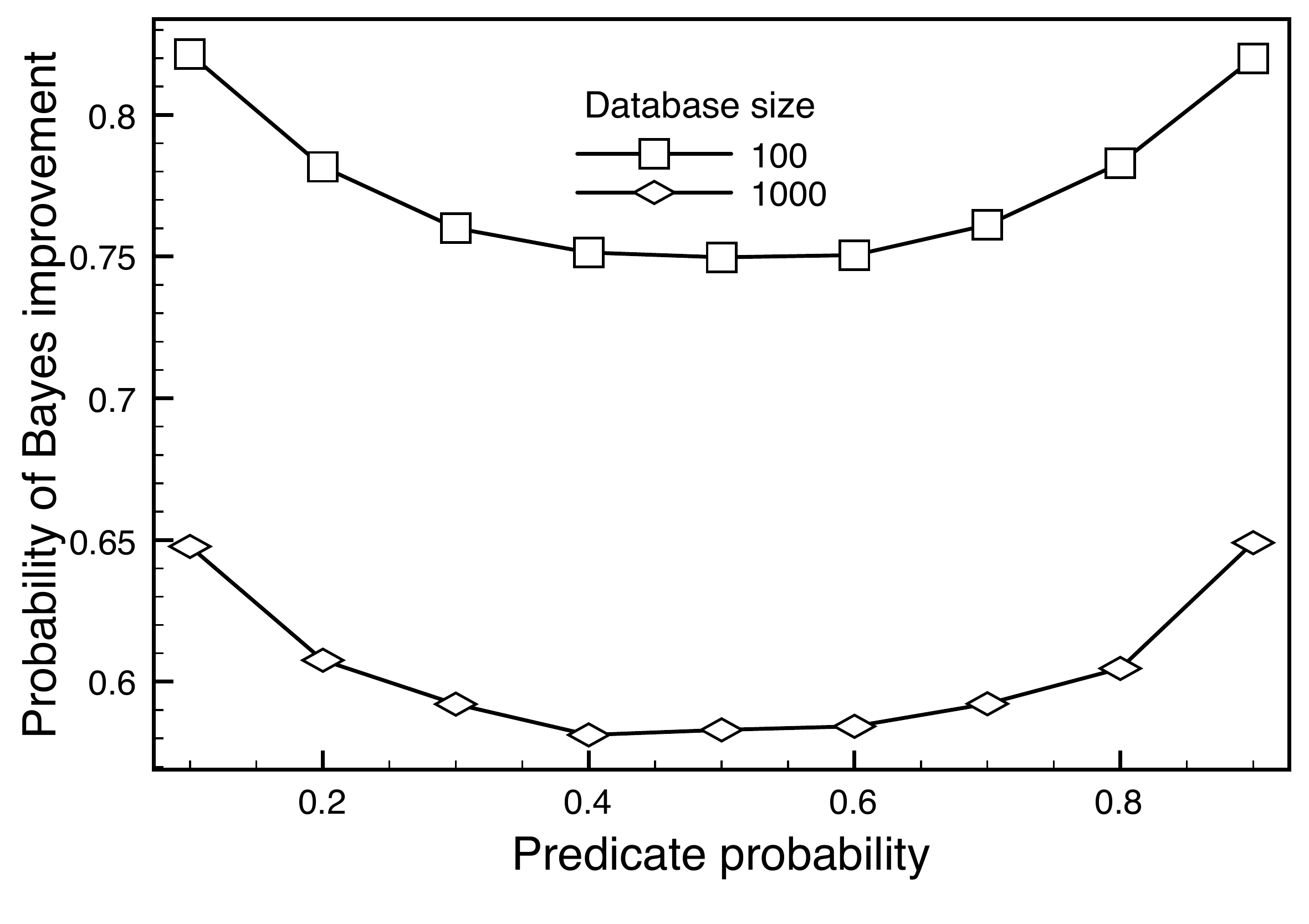}
    \caption{Probability of Bayes improvement vs predicate probability}
    \label{fig:Perr-3}
\end{center}
\end{figure}

\section{Conclusions}
Adding Laplacian noise to a response (an implementation of the output perturbation approach) has been shown to achieve $\epsilon$-differential privacy in a statistical database. This approach establishes a trade-off between accuracy and privacy. Under this approach, the noisy response delivered by the database owner can be seen as a (\textit{naive}) estimate of the true response. However, the user who has submitted the query may try to improve that estimate to achieve a better accuracy.

We have shown that, by using a single noisy response, the user may adopt a Bayesian approach and obtain a refined estimate of the true response to its query. The additional information required by the Bayesian approach is the database size and the probability that the predicate incorporated in the query is satisfied. We have investigated the characteristics of the resulting estimate comparing it to the naive approach through two performance metrics: the average error and the probability of getting a lower error.

We have found that the Bayesian approach provides a better estimate in all cases. The improvement gets larger as the differential privacy requirement gets tighter, with the average error growing linearly for the naive estimator and levelling off for the Bayes one.

The introduction of the Bayes estimate allows therefore to alter the accuracy-privacy trade-off resulting from noise addition, providing a better accuracy for the same level of noise addition (hence, the same level of differential privacy). 


\end{document}